\documentclass{article}

\usepackage{footnote}
\usepackage[top=3cm,left=3cm,right=3cm,bottom=3cm]{geometry}
\usepackage{complexity}
\usepackage[top=3cm,left=3cm,right=3cm,bottom=3cm]{geometry}
\usepackage[intlimits]{amsmath}
\usepackage{amsfonts,amssymb,amsthm,bbm, mathtools}
\usepackage[colorlinks=true]{hyperref}
\usepackage{verbatim}
\usepackage{enumerate}
\usepackage[vlined,ruled, linesnumbered]{algorithm2e}
\usepackage{xcolor}

\usepackage{thmtools}
\usepackage{thm-restate}

\hypersetup{pdfstartview=FitH}

\newtheorem{theorem}{Theorem}

\newtheorem{fact}[theorem]{Fact}

\theoremstyle{definition}
\newtheorem{definition}{Definition}

\newtheorem{proposition}[theorem]{Proposition}
\newtheorem{claim}[theorem]{Claim}

\newclass{\alt}{alt}
\newclass{\s}{s}
\newclass{\bs}{bs}
\newclass{\fbs}{fbs}
\newclass{\fC}{fC}
\newclass{\Cert}{C}
\newclass{\EC}{EC}
\newclass{\SB}{sb}
\newclass{\fsb}{fsb}
\newclass{\sC}{sC}
\newclass{\qa}{q_{adv}}
\newclass{\disc}{disc}
\newclass{\bias}{bias}
\newclass{\snip}{snip}
\newclass{\Snip}{Snip}
\newclass{\codim}{codim}

\newcommand{\clB}{\mathcal{B}}
\newcommand{\clC}{\mathcal{C}}

\newcommand{\email}[1]{\href{mailto:#1}{#1}}

\allowdisplaybreaks
\title{A Composition Theorem for Randomized Query complexity}
\author{Anurag Anshu \thanks{Centre for Quantum Technologies, National University of Singapore, Singapore. \email{a0109169@u.nus.edu}} \and 
Dmitry Gavinsky \thanks{Institute of Mathematics, Czech Academy of Sciences, \v Zitna 25, Praha 1, Czech Republic. Part of this work was done when Dmitry Gavinsky was visiting the Centre for Quantum Technologies at the National University of Singapore.} \and
Rahul Jain \thanks{Centre for Quantum Technologies, National University of Singapore and MajuLab, UMI 3654, Singapore. \email{rahul@comp.nus.edu.sg}} \and 
Srijita Kundu  \thanks{Centre for Quantum Technologies, National University of Singapore, Singapore. \email{srijita.kundu@u.nus.edu}} \and 
Troy Lee \thanks{Division of Mathematical Sciences, Nanyang Technological University, Singapore and Centre for Quantum Technologies, National University of Singapore, Singapore. \email{troyjlee@gmail.com}} \and
Priyanka Mukhopadhyay \thanks{Centre for Quantum Technologies, National University of Singapore, Singapore. \email{a0109168@u.nus.edu}}  \and
Miklos Santha \thanks{IRIF, Universit\'e Paris Diderot, CNRS, 75205 Paris, France, and Centre for Quantum Technologies, National University of Singapore, Singapore. \email{santha@irif.fr}
}\and
Swagato Sanyal \thanks{Division of Mathematical Sciences, Nanyang Technological University, Singapore and Centre for Quantum Technologies, National University of Singapore, Singapore. \email{ssanyal@ntu.edu.sg}}}

\begin{document}
\maketitle
\begin{abstract}
Let the randomized query complexity of a relation for error probability $\epsilon$ be denoted by $\R_\epsilon(\cdot)$. We prove that for any relation $f \subseteq \{0,1\}^n \times \mathcal{R}$ and Boolean function $g:\{0,1\}^m \rightarrow \{0,1\}$,  $\R_{1/3}(f\circ g^n) = \Omega(\R_{4/9}(f)\cdot\R_{1/2-1/n^4}(g))$, where $f \circ g^n$ is the relation obtained by composing $f$ and $g$. We also show using an XOR lemma that $\R_{1/3}\left(f \circ \left(g^\oplus_{O(\log n)}\right)^n\right)=\Omega(\log n \cdot \R_{4/9}(f) \cdot \R_{1/3}(g))$, where $g^\oplus_{O(\log n)}$ is the function obtained by composing the XOR function on $O(\log n)$ bits and $g$.
\end{abstract}
\section{Introduction}
Given two Boolean functions $f:\{0,1\}^n \rightarrow \{0,1\}$ and $g:\{0,1\}^m \rightarrow \{0,1\}$, the composed function $f\circ g^n:\left(\{0,1\}^m\right)^n \rightarrow \{0,1\}$ is defined as follows: For $x=(x^{(1)}, \ldots, x^{(n)}) \in \left(\{0,1\}^m\right)^n$, $f \circ g^n(x)=f(g(x^{(1)}), \ldots, g(x^{(n)}))$. Composition of Boolean functions has long been a topic of active research in complexity theory. In many works, composition of Boolean function is studied in the context of a certain complexity measure. The objective is to understand the relation between the complexity of the composed function in terms of the complexities of the individual functions.
Let $\D(\cdot)$ denote the deterministic query complexity. It is easy to see that $\D(f \circ g^n) \leq \D(f) \cdot \D(g)$ since $f\circ g$ can be computed by simulating an optimal query algorithm of $f$; whenever the algorithm makes a query, we simulate an optimal query algorithm of $g$ and serve the query. It can be shown by an adversary argument that this is an optimal query algorithm and $\D(f \circ g^n) = \D(f) \cdot \D(g)$.

However, such a characterization is not so obvious for randomized query complexity. Although a similar upper bound still holds true (possibly accommodating a logarithmic overhead), it is no more as clear that it also asymptotically bounds the randomized query complexity of $f\circ g^n$ from below. Let $\R_\epsilon(\cdot)$ denote the $\epsilon$-error randomized query complexity. Our main theorem in this work is the following.

\begin{restatable}[Main Theorem]{thm}{composing}
\label{thm:composing}
For any relation $f \subseteq \{0,1\}^n \times \mathcal{R}$ and Boolean function $g: \{0,1\}^m \to \{0,1\}$,
\[ \R_{1/3}(f\circ g^n) = \Omega(\R_{4/9}(f)\cdot\R_{1/2-1/n^4}(g)).\]
\end{restatable}
See Section~\ref{prelim} for definitions of composition and various complexity measures of relations. Theorem~\ref{thm:composing} implies that if $g$ is a function that is hard to compute with error $1/2-1/n^4$, $f \circ g^n$ is hard to compute with error $1/3$.

In the special case where $f$ is a Boolean function, Theorem~\ref{thm:composing} implies that $\R_{1/3}(f\circ g^n) = \Omega(\R_{1/3}(f)\cdot\R_{1/2-1/n^4}(g))$, since the success probability of query algorithms for Boolean functions can be boosted from $5/9$ to $2/3$ by constantly many independent repetitions followed by taking a majority of the different outputs.

Theorem~\ref{thm:composing} is useful only when the function $g$ is hard against randomized query algorithms even for error $1/2-1/n^4$. In Section~\ref{directsum} we prove the following consequence of Theorem~\ref{thm:composing}.

Let $f \subseteq \{0,1\}^n \times \mathcal{R}$ be any relation. Let $g:\{0,1\}^m \rightarrow \{0,1\}$ be a function. Let $g_t^\oplus: \left(\{0,1\}^m\right)^t \rightarrow \{0,1\}$ be defined as follows: for $x=(x^{(1)}, \ldots, x^{(t)}) \in \left(\{0,1\}^m\right)^t$, $g_t^\oplus(x) = \oplus_{i=1}^tg(x^{(i)})$.
\begin{restatable}{thm}{xor}
\label{thm:xor}
\[\R_{1/3}\left(f \circ \left(g^\oplus_{O(\log n)}\right)^n\right)=\Omega(\log n \cdot \R_{4/9}(f) \cdot \R_{1/3}(g)).\]
\end{restatable}
Theorem~\ref{thm:xor} is proved by establishing, via an XOR lemma by Andrew Drucker \cite{DBLP:conf/coco/Drucker11}, that if $g$ is hard for error $1/3$ then $g_{O(\log n)}^\oplus$ is hard for error $1/2-1/n^4$.

Composition theorem for randomized query complexity has been an area of active research in the past. G{\"{o}}{\"{o}}s and Jayram \cite{DBLP:conf/coco/GoosJ16} showed a composition theorem for a constrained version of conical junta degree, which is a lower bound on randomized query complexity. Composition theorem for approximate degree (which also lower bounds randomized query complexity) for the special case of TRIBES function has seen a long line of research culminating in independent works of Sherstov \cite{DBLP:journals/toc/Sherstov13a} and Bun and Thaler \cite{DBLP:conf/icalp/BunT13} who settle the question by proving optimal bounds.

Composition theorem has been studied and shown in the context of communication and query complexities by the works of G{\"{o}}{\"{o}}s, Pitassi and Watson \cite{DBLP:conf/focs/GoosP015,DBLP:journals/corr/GoosP017}, Chattopadhyay et al. \cite{DBLP:journals/corr/ChattopadhyayKL17} when the function $g$ is the indexing function or the inner product function with large enough arity. The work of Hatami, Hosseini and Lovett \cite{DBLP:conf/focs/HatamiHL16} proves a composition theorem in the context of communication and parity query complexites when the function $g$ is the two-bit XOR function. Ben-David and Kothari \cite{DBLP:conf/icalp/Ben-DavidK16} proved a composition theorem for the \emph{sabotage complexity} of Boolean functions, a novel complexity measure defined in the same work that the authors prove to give quadratically tight bound on the randomized query complexity.

Composition theorems have also been successfully used in the past in constructing separating examples for various complexity measures, and bounding one complexity measure in terms of another. Kulkarni and Tal \cite{DBLP:journals/cjtcs/KulkarniT16} proved an upper bound on fractional block sensitivity in terms of degree by analyzing the behavior of fractional block sensitivity under function composition. Separation between block sensitivity and degree was obtained by composing Kushilevitz's icosadedron function repeatedly with itself (see \cite{DBLP:journals/tcs/BuhrmanW02}). Separation between parity decision tree complexity and Fourier sparsity has been obtained by O'Donnell et al. by studying the behavior of \emph{parity kill number} under function composition \cite{DBLP:conf/coco/ODonnellWZST14}.
\subsection{Our techniques}
\label{technique}
In this section, we give a high level overview of our proof of Theorem~\ref{thm:composing}. We refer the reader to Section~\ref{prelim} for formal definitions of composition and various complexity measures of relations.

Let $\epsilon=1/2-1/n^4$. Let $\mu$ be the distribution over the domain $\{0,1\}^m$ of $g$ for which $\R_\epsilon(g)$ is achieved, i.e., $\R_\epsilon(g)=\D^\mu_\epsilon(g)$ (see Fact~\ref{minmax}). For $b \in \{0,1\}$, let $\mu_b$ denote the distribution obtained by conditioning $\mu$ to the event that $g(x)=b$ (see Section~\ref{prelim} for a formal definition).

We show that for every probability distribution $\lambda$ over the domain $\{0,1\}^n$ of $f$, there exists a deterministic query algorithm $\mathcal{A}$ with worst case query complexity at most $\R_{1/3}(f \circ g^n)/\R_\epsilon(g)$, such that $\Pr_{z \sim \lambda} [(z,\mathcal{A}(z))\in f] \geq 5/9$. By the minimax principle (Fact~\ref{minmax}) this proves Theorem~\ref{thm:composing}.

Now using the distribution $\lambda$ over $\{0,1\}^n$ we define a probability distribution $\gamma$ over $\left(\{0,1\}^m\right)^n$. To define $\gamma$, we begin by defining a family of distributions $\{\gamma^z: z \in \{0,1\}^n\}$ over $\left(\{0,1\}^m\right)^n$. For a fixed $z=(z_1, \ldots, z_n) \in \{0,1\}^n$, we define $\gamma^z$ by giving a sampling procedure:
\begin{enumerate}
\item For each $i=1, \ldots, n$, sample $x^{(i)}=(x^{(i)}_1, \ldots, x^{(i)}_m)$ from $\{0,1\}^m$ independently according to $\mu_{z_i}$.
\item Return $x=(x^{(1)}, \ldots, x^{(n)})$.
\end{enumerate}
Thus for $z=(z_1, \ldots, z_n) \in \{0,1\}^n$ and $x=(x^{(1)}, \ldots, x^{(n)}) \in (\{0,1\}^m)^n$, $\gamma^z(x)=\Pi_{i=1}^n \mu_{z_i}(x^{(i)})$. Note that $\gamma^z$ is supported only on strings $x$ for which the following is true: for each $r \in \mathcal{R}$, $(x,r) \in f \circ g^n$ if and only if $(z,r) \in f$.

Having defined the distributions $\gamma^z$, we define the distribution $\gamma$ by giving a sampling procedure:
\begin{enumerate}
\item Sample a $z=(z_1, \ldots,z_n)$ from $\{0,1\}^n$ according to $\lambda$.
\item Sample an $x=(x^{(1)}, \ldots, x^{(n)})$ from $(\{0,1\}^m)^n$ according to $\gamma^z$. Return $x$.
\end{enumerate}

By minimax principle (Fact~\ref{minmax}), there is a deterministic query algorithm $\mathcal{B}$ of worst case complexity at most $\R_{1/3} (f \circ g^n)$ such that $\Pr_{x \sim \gamma}[(x,\mathcal{B}(x)) \in f \circ g^n] \geq 2/3$. We will use $\mathcal{B}$ to construct a randomized query algorithm $\mathcal{A'}$ for $f$ with the desired properties. A deterministic query algorithm $\mathcal{A}$ for $f$ with required performance guarantees can then be obtained by appropriately fixing the randomness of $\mathcal{A}'$.

See Algorithm~\ref{A'} for a formal description of $\mathcal{A'}$. Given an input $z=(z_1, \ldots, z_n), \mathcal{A}'$ simulates $\mathcal{B}$. Recall that an input to $\mathcal{B}$ is an $nm$ bit long string $(x^{(i)}_j)_{{i=1, \ldots, n}\atop{ j= 1, \ldots, m}}$. Whenever $\mathcal{B}$ asks for (queries) an input bit $x^{(i)}_j$, a response bit is appropriately generated and passed to $\mathcal{B}$. To generate a response to a query by $\mathcal{B}$, a bit in $z$ may be queried; those queries will contribute to the query complexity of $\mathcal{A}'$. The queries are addressed as follows. Let the simulation of $\mathcal{B}$ request bit $x_j^{(i)}$.
\begin{itemize}
\item If less than $\D_\epsilon^\mu(g)$ queries have been made into $x^{(i)}$ (including the current query) then a bit $b$ is sampled from the marginal distribution of $x_j^{(i)}$ according to $\mu$, conditioned on the responses to the past queries. $b$ is passed to the simulation of $\mathcal{B}$.
\item If $\D_\epsilon^\mu(g)$ queries have been made into $x^{(i)}$ (including the current query) then first the input bit $z_i$ is queried; then a bit $b$ is sampled from the marginal distribution of $x_j^{(i)}$ according to $\mu_{z_i}$, conditioned on the responses to the past queries. $b$ is passed to the simulation of $\mathcal{B}$.
\end{itemize}
The simulation of $\mathcal{B}$ continues until $\mathcal{B}$ terminates in a leaf. Then $\mathcal{A}'$ also terminates and outputs the label of the leaf.

We use Claims~\ref{R->bias} and~\ref{unbias} to prove that for a fixed $z \in \{0,1\}^n$, the probability distribution induced by $\mathcal{A}'$ on the leaves of $\mathcal{B}$ is statistically close to the probability distribution induced by $\mathcal{B}$ on its leaves for a random input from $\gamma^z$. Averaging over different $z$'s, the correctness of $\mathcal{A}'$ follows from the correctness of $\mathcal{B}$. The reader is referred to Section~\ref{mainpf} for the details.

\section{Preliminaries}
\label{prelim}
In this section, we define some basic concepts, and set up our notations. We begin with defining the 2-sided error randomized and distributional query complexity measures of relations. The relations considered in this work will all be between the Boolean hypercube $\{0,1\}^k$ of some dimension $k$, and an arbitrary set $\mathcal{S}$. The strings $x \in \{0,1\}^n$ will be called as inputs to the relation, and $\{0,1\}^n$ will be referred to as the \emph{input space} and the \emph{domain} of $h$.
\begin{definition}[2-sided Error Randomized Query Complexity]
Let $\mathcal{S}$ be any set. Let $h \subseteq \{0,1\}^k \times \mathcal{S}$ be any relation and $\epsilon \in [0,1/2)$. The 2-sided error randomized query complexity $\R_\epsilon(h)$ is the minimum number of queries made in the worst case by a randomized query algorithm $\mathcal A$ (the worst case is over inputs and the internal randomness of $\mathcal{A}$) that on each input $x \in \{0,1\}^k$ satisfies $\Pr[(x,\mathcal A(x)) \in h] \geq 1 - \epsilon$ (where the probability is over the internal randomness of $\mathcal{A}$). 
\end{definition}

\begin{definition}[Distributional Query Complexity]
 Let $h \subseteq \{0,1\}^k \times \mathcal{S}$ be any relation, $\mu$ a distribution on the input space $\{0,1\}^k$ of $h$, and $\epsilon \in [0,1/2)$. The distributional query complexity $\D^\mu_\epsilon(h)$ is the minimum number of queries made in the worst case (over inputs) by a deterministic query algorithm $\mathcal A$ for which $\Pr_{x \sim \mu}[(x,\mathcal A(x)) \in h] \geq 1 - \epsilon$.
\end{definition}

In particular, if $h$ is a function and $\mathcal{A}$ is a randomized or distributional query algorithm computing $h$ with error $\epsilon$, then $\Pr [h(x)=\mathcal{A}(x)] \geq 1-\epsilon$, where the probability is over the respective sources of randomness.

The following theorem is von Neumann's minimax principle stated for decision trees.
\begin{fact}[minimax principle]
\label{minmax}
For any integer $k$, set $\mathcal{S}$, and relation $h \subseteq \{0,1\}^k \times \mathcal{S}$,
\[\R_\epsilon(h)=\max_{\mu}\D_\epsilon^\mu(h).\]
\end{fact}

Let $g:\{0,1\}^m \rightarrow \{0,1\}$ be a Boolean function. Let $\mu$ be a probability distribution on $\{0,1\}^m$ which intersects non-trivially both with $g^{-1}(0)$ and with $g^{-1}(1)$. For each $z \in \{0,1\}$, let $\mu_z$ be the distribution obtained by restricting $\mu$ to $g^{-1}(z)$. Formally,
\[\mu_z(x)=\left\{  \begin{array}{ll} $0$ & \mbox{if $g(x) \neq z$} \\
\frac{\mu(x)}{\sum_{y: g(y)=z} \mu(y)} & \mbox{if $g(x)=z$}\end{array}   \right.\]
Notice that $\mu_0$ and $\mu_1$ are defined with respect to some Boolean function $g$, which will always be clear from the context.

\begin{definition}[Subcube, Co-dimension]
A subset $\clC$ of $\{0,1\}^m$ is called a subcube if there exists a set $S \subseteq \{1, \ldots, m\}$ of indices and an \emph{assignment function} $A:S \rightarrow \{0,1\}$ such that $\clC=\{x \in \{0,1\}^m:\forall i \in S, x_i=A(i)\}$. The co-dimension $\codim(\clC)$ of $\clC$ is defined to be $|S|$. 
\end{definition}
Let $\clC \subseteq \{0,1\}^m$ be a subcube and $\mu$ be a probability distribution on $\{0,1\}^m$. We will often abuse notation and use $\clC$ to denote the event that a random string $x$ belongs to the subcube $\clC$. The probability $\Pr_{x \sim \mu}[x \in \clC]$ will be denoted by $\Pr_\mu[\clC]$. For subcubes $\clC_1$ and $\clC_2$, the conditional probability $\Pr_{x \sim \mu}[x \in \clC_2 \mid x \in \clC_1]$ will be denoted by $\Pr_\mu[\clC_2 \mid \clC_1]$.
\begin{definition}[Bias of a subcube]
Let $g:\{0,1\}^m \rightarrow \{0,1\}$ be a Boolean function. Let $\mu$ be a probability distribution over $\{0,1\}^m$. Let $\clC \subseteq \{0,1\}^m$ be a subcube such that $\Pr_{\mu}[\clC] > 0$. The bias of $\clC$ with respect to $\mu$, $\bias^\mu(\clC)$, is defined to be:
\[\bias^\mu(\clC)=|\Pr_{x \sim \mu}[g(x)=0 \mid x \in \clC ]-\Pr_{x \sim \mu}[g(x)=1 \mid x \in \clC ]|.\]
A Boolean function $g$ is implicit in the definition of bias, which will always be clear from the context.
\end{definition}
\begin{proposition}
\label{fullbias}
Let $g:\{0,1\}^m \rightarrow \{0,1\}$ be a Boolean function, and $\D_{\epsilon}^\mu(g) > 0$. Then,
\[\min_{b \in \{0,1\}}\{\Pr_{x \sim \mu}[g(x)=b]\} > \epsilon.\] In particular, $\bias^\mu(\{0,1\}^m) < 1-2\epsilon$.
\end{proposition}
\begin{proof}
Towards a contradiction, assume that $\min_{b \in \{0,1\}}\{\Pr_{x \sim \mu}[g(x)=b]\} \leq \epsilon$. Then, the algorithm that outputs  $\arg \max_{b \in \{0,1\}}\{\Pr_{x \sim \mu}[g(x)=b]\}$ makes 0 query and is correct with probability at least $1-\epsilon$. This contradicts the hypothesis that $\D_{\epsilon}^\mu(g) >0$.
\end{proof}
Now we define composition of two relations.
\begin{definition}[Composition of relations]
\label{def:comp}
Let $f \subseteq \{0,1\}^n \times \mathcal{R}$ and $g \subseteq \{0,1\}^m \times \{0,1\}$ be two relations. The composed relation $f \circ g^n \subseteq \left(\{0,1\}^m\right)^n \times \mathcal{R}$ is defined as follows: For $x=(x^{(1)}, \ldots, x^{(n)}) \in \left(\{0,1\}^m\right)^n$ and $r \in \mathcal{R}$, $(x,r) \in f \circ g^n$ if and only if there exists $b=(b^{(1)}, \ldots, b^{(n)}) \in \{0,1\}^n$ such that for each $i=1, \ldots, n$, $(x^{(i)},b^{(i)}) \in g$ and $(b,r) \in f$.
\end{definition}
We will often view a deterministic query algorithm as a binary decision tree. In each vertex $v$ of the tree, an input variable is queried. Depending on the outcome of the query, the computation goes to a child of $v$. The child of $v$ corresponding to outcome $b$ to the query made is denoted by $v_b$.
It is well known that the set of inputs that lead the computation of a decision tree to a certain vertex forms a subcube. We will denote the subcube corresponding to a vertex $v$ by $\clC_v$.

We next prove two claims about bias, probability and co-dimension of subcubes that will be useful. Claim~\ref{R->bias} states that for a function with large distributional query complexity, the bias of most shallow leaves of any deterministic query procedure is small.
\begin{claim}
\label{R->bias}
Let $g:\{0,1\}^m \rightarrow \{0,1\}$ be a Boolean function. Let $\epsilon \in [1/4, 1/2)$ and let $\delta=1/2-\epsilon$. Let $\mu$ be a probability distribution on $\{0,1\}^m$, and $\D_\epsilon^\mu(g)=c > 0$. Let $\mathcal{B}$ be any deterministic query algorithm for strings in $\{0,1\}^m$. For each $y \in \{0,1\}^m$, let $\ell_y$ be the unique leaf of $\mathcal{B}$ that contains $y$. Then,
\begin{enumerate}
\item[(a)] $\Pr_{y \sim \mu} [\mathsf{codim}(\ell_y) < c\mbox{\ and\ }\bias^\mu(\ell_y)\geq2\delta^{1/2}] < \delta^{1/2}.$
\item[(b)] For each $b \in \{0,1\}$, $\Pr_{y \sim \mu_b} [\mathsf{codim}(\ell_y) < c\mbox{\ and\ }\bias^\mu(\ell_y)\geq2\delta^{1/2}] < 4\delta^{1/2}$.
\end{enumerate}
\end{claim}
In the above claim $\mathcal{B}$ could just be a deterministic procedure that makes queries and eventually terminates; whether or not it makes any output upon termination is not of any consequence here.
\begin{proof}
We first show that part (a) implies part (b). To this end, assume part (a) and fix a $b \in \{0,1\}$. Let $a(y)$ be the indicator variable for the event $\mathsf{codim}(\ell_y) < c$ and $\bias^\mu(\ell_y)\geq2\delta^{1/2}$. Thus, part (a) states that $\Pr_{y \sim \mu}[a(y)=1] < \delta^{1/2}$. Now,
\begin{align*}
&\Pr_{y \sim \mu_b} [\mathsf{codim}(\ell_y) < c\mbox{\ and\ }\bias^\mu(\ell_y)\geq2\delta^{1/2}] \\
&=\sum_{y:a(y)=1} \mu_b(y) \\
&=\frac{1}{\sum_{y :g(y)=b}\mu(y)}\sum_{y: a(y)=1} \mu(y)\mbox{\ \ (From the definition of $\mu_b$)} \\
& < \frac{1}{\epsilon} \Pr_{y \sim \mu}[a(y)=1] \mbox{\ \ (From Proposition~\ref{fullbias})}\\
&< 4\delta^{1/2}. \mbox{\ \ (By the hypothesis $\epsilon \geq 1/4$ and part (a))}
\end{align*}

We now prove part (a). Towards a contradiction assume that \[\Pr_{y \sim \mu} [\mathsf{codim}(\ell_y) < c\mbox{\ and\ }\bias^\mu(\ell_y)\geq2\delta^{1/2}] \geq \delta^{1/2}.\] Now consider the following decision tree algorithm $\mathcal{A}$ on $m$ bit strings:

Begin simulating $\mathcal{B}$. Let $\mathcal{C}$ be the subcube associated with the current node of $\mathcal{B}$ in the simulation. Simulate $\mathcal{B}$ unless one of the following happens.
\begin{itemize}
\item $\mathcal{B}$ terminates.
\item The number of queries made is $c-1$.
\item $\bias^\mu(\mathcal{C}) \geq 2\delta^{1/2}$.
\end{itemize}
Upon termination, if $\bias^\mu(\mathcal{C}) \geq 2\delta^{1/2}$, output $\arg \max_{b \in \{0,1\}} \Pr_{y \sim \mu}[g(y)=b \mid y \in \mathcal{C}]$. Else output a uniformly random bit.

It immediately follows that the worst case query complexity of $\mathcal{A}$ is at most $c-1$. Now, we will prove that $\Pr_{y \sim \mu}[\mathcal{A}(y)=g(y)] \geq 1-\epsilon$. This will contradict the hypothesis that $\D_\epsilon^\mu(g)=c$. Let $\mathcal{L}$ be the node of $\mathcal{B}$ at which the computation of $\mathcal{A}$ ends. Let $\Pr_{y \sim \mu}[\bias^\mu (\mathcal{L}) \geq 2 \delta^{1/2}]=p$. By our assumption, the probability (over $\mu$) that $\mathcal{L}$ is a leaf and $\bias^\mu (\mathcal{L}) \geq 2 \delta^{1/2}$ is at least $\delta^{1/2}$; in particular $p \geq \delta^{1/2}$. Now,
\begin{align*}
&\Pr_{y \sim \mu}[\mathcal{A}(y)=g(y)] \\
& =\Pr_{y \sim \mu}[\bias^\mu (\mathcal{L}) \geq 2 \delta^{1/2}]\cdot\Pr_{y \sim \mu}[\mathcal{A}(y)=g(y) \mid \bias^\mu (\mathcal{L}) \geq 2 \delta^{1/2}]+\\
&\qquad \qquad \qquad \Pr_{y \sim \mu}[\bias^\mu (\mathcal{L}) < 2 \delta^{1/2}]\cdot\Pr_{y \sim \mu}[\mathcal{A}(y)=g(y) \mid \bias^\mu (\mathcal{L}) < 2 \delta^{1/2}] \\
&\geq p \cdot (1/2+\delta^{1/2}) + (1-p).\frac{1}{2}\mbox{\ \  \ (from our assumption)}\\
&=1/2 + p\cdot \delta^{1/2}\\
&\geq 1/2+\delta\mbox{\ \ (since $p \geq \delta^{1/2}$)}\\
&=1-\epsilon.
\end{align*}
This completes the proof.
\end{proof}
The next claim states that if a subcube has low bias with respect to a distribution $\mu$, then the distributions $\mu_0$ and $\mu_1$ ascribe almost the same probability to it.
\begin{claim}
\label{unbias}
Let $g:\{0,1\}^m \rightarrow \{0,1\}$ be a Boolean function and $\delta \in (0,\frac{1}{2}]$. Let $\mu$ be a distribution on $\{0,1\}^m$. Let $\clC$ be a subcube such that $\Pr_{\mu} [ \clC]>0$ and $\bias^\mu(\clC) \leq \delta$. Also assume that $\bias^\mu(\{0,1\}^m) \leq \delta$. Then for any $b \in \{0,1\}$ we have,
\begin{enumerate}
\item[(a)] $\Pr_{\mu} [\clC] \leq (1+4\delta)\cdot\Pr_{\mu_b} [\clC],$
\item[(b)] $\Pr_{\mu} [\clC] \geq (1-4\delta)\cdot\Pr_{\mu_b} [ \clC]$.
\end{enumerate}
\end{claim}
\begin{proof}
We prove part (a) of the claim. The proof of part (b) is similar.

By the definition of $\bias$ and the hypothesis, for each $b \in \{0,1\}$,
\begin{align}
\sum_{y \in \mathcal H_m: g(y)=b} \mu(y) \leq \left(\frac{1}{2}+\frac{\delta}{2}\right)\cdot \sum_{y \in \mathcal H_m} \mu(y)=\frac{1}{2}+\frac{\delta}{2}, \label{oneone}\\
\sum_{y \in \clC: g(y)=b} \mu(y) \geq \left(\frac{1}{2}-\frac{\delta}{2}\right) \cdot\sum_{y \in \clC} \mu(y)  > 0. \label{two}
\end{align}
Now,
\begin{align*}
&\Pr_{\mu_b} [\clC]=\sum_{y \in \clC} \mu_b(y) \\
&=\frac{ \sum_{y \in \clC:g(y)=b}\mu(y)}{\sum_{y \in \mathcal H_m : g(y)=b} \mu(y)} \\
& \geq \frac{(1/2 - \delta/2)\cdot\sum_{y \in \clC} \mu(y)}{1/2+\delta/2}  \mbox{\ \ \ (From Equations~(\ref{oneone}) and~(\ref{two})}\\
& = \frac{1/2 - \delta/2}{1/2+\delta/2} \cdot \Pr_{\mu} [\clC]
\end{align*}
Thus,
\begin{align*}
\Pr_{\mu} [\clC] &\leq \frac{1/2+\delta/2}{1/2-\delta/2} \cdot \Pr_{\mu_b} [\clC] \leq (1+4\delta) \cdot\Pr_{\mu_b} [\clC].\mbox{\ \ \ \ (since $\delta \leq \frac{1}{2})$}
\end{align*}
\end{proof}
\section{Composition Theorem}
\label{mainpf}
In this section we prove our main theorem. We restate it below.
\composing*
\begin{proof}We begin by recalling the notations defined in Section~\ref{technique} that we will use in this proof.

Let $\epsilon=1/2-1/n^4$. Let $\mu$ be the distribution over the domain $\{0,1\}^m$ of $g$ for which $\R_\epsilon(g)$ is achieved, i.e., $\R_\epsilon(g)=\D^\mu_\epsilon(g)$. (see Fact~\ref{minmax})

We show that for every probability distribution $\lambda$ over the input space $\{0,1\}^n$ of $f$, there exists a deterministic query algorithm $\mathcal{A}$ with worst case query complexity at most $\R_{1/3}(f \circ g)/\R_\epsilon(g)$, such that $\Pr_{z \sim \lambda} [(z,\mathcal{A}(z))\in f] \geq 5/9$. By the minimax principle (Fact~\ref{minmax}) this will prove Theorem~\ref{thm:composing}.

Using $\lambda$, we define a probability distribution $\gamma$ over $\left(\{0,1\}^m\right)^n$. We first define a family of distributions $\{\gamma^z: z \in \{0,1\}^n\}$ over $\left(\{0,1\}^m\right)^n$. For a fixed $z \in \{0,1\}^n$, we define $\gamma^z$ by giving a sampling procedure:
\begin{enumerate}
\item For each $i=1, \ldots, n$, sample $x^{(i)}=(x^{(i)}_1, \ldots, x^{(i)}_m)$ from $\{0,1\}^m$ independently according to $\mu_{z_i}$.
\item Return $x=(x^{(1)}, \ldots, x^{(n)})$.
\end{enumerate}
Thus for $z=(z_1, \ldots, z_n) \in \{0,1\}^n$ and $x=(x^{(1)}, \ldots, x^{(n)}) \in (\{0,1\}^m)^n$, $\gamma^z(x)=\Pi_{i=1}^n \mu_{z_i}(x^{(i)})$. Note that $\gamma^z$ is supported only on strings $x$ for which the following is true: for each $r \in \mathcal{R}$, $(x,r) \in f \circ g^n$ if and only if $(z,r) \in f$.

Now, we define the distribution $\gamma$ by giving a sampling procedure:
\begin{enumerate}
\item Sample a $z=(z_1, \ldots,z_n)$ from $\{0,1\}^n$ according to $\lambda$.
\item Sample an $x=(x^{(1)}, \ldots, x^{(n)})$ from $(\{0,1\}^m)^n$ according to $\gamma^z$. Return $x$.
\end{enumerate}

By the minimax principle (Fact~\ref{minmax}), there is a deterministic query algorithm $\mathcal{B}$ of worst case complexity at most $\R_{1/3} (f \circ g^n)$ such that $\Pr_{x \sim \gamma}[(x,\mathcal{B}(x)) \in f \circ g^n] \geq 2/3$. We will use $\mathcal{B}$ to construct a randomized query algorithm $\mathcal{A}'$ for $f$ with the desired properties. A deterministic query algorithm $\mathcal{A}$ for $f$ with required performance guarantees can then be obtained by appropriately fixing the randomness of $\mathcal{A}'$. Algorithm~\ref{A'} formally defines the algorithm $\mathcal{A}'$ that we construct.

\begin{algorithm}[!h]\label{A'}
\DontPrintSemicolon
\caption{Randomized query algorithm $\mathcal{A}'$ for $f$}
\KwIn{$z \in \{0,1\}^n$}
Initialize $v$ $\gets$ root of the decision tree $\clB$, $Q \gets \emptyset$ \;
\While{$v$ is not a leaf}
{ Let a bit in $x^{(i)}$ be queried at $v$ \;
  \If(\tcc*[f]{$\codim(\clC_v^{(i)}) < \D^\mu_\epsilon(g)$ if this is satisfied}){$i \not \in Q$} 
   {Set $v \gets v_b$ with probability $\Pr_\mu[\clC^{(i)}_{v_b}\mid\clC^{(i)}_v]$ \; \label{cond1}
   \If{$\codim(\clC_v^{(i)}) = \D^\mu_\epsilon(g)$}
    {Query $z_i$ \;
     Set $Q = Q \cup \{i\}$
    }
   }
   \Else
   {Set $v \leftarrow v_b$ with probability $\Pr_{\mu_{z_i}}[\clC^{(i)}_{v_b}\mid\clC^{(i)}_v]$ \label{cond2}
   }
 }
Output label of $v$.
\end{algorithm}

From the definition of $\bias$ one can verify that the events in steps~\ref{cond1} and~\ref{cond2} in Algorithm~\ref{A'} that are being conditioned on, have non-zero probabilities under the respective distributions; hence, the probabilistic processes are well-defined.

From the description of $\mathcal{A}'$ it is immediate that $z_i$ is queried only if the underlying simulation of $\mathcal{B}$ queries at least $\R_\epsilon(g)$ locations in $x^{(i)}$. Thus the worst-case query complexity of $\mathcal{A}'$ is at most $\R_{1/3} (f \circ g^n)/\R_\epsilon(g)$.

We are left with the task of bounding the error of $\mathcal{A}'$. Let $\mathcal{L}$ be the set of leaves of the decision tree $\mathcal{B}$. Each leaf $\ell \in \mathcal{L}$ is labelled with a bit $b_\ell \in \{0,1\}$; whenever the computation reaches $\ell$, the bit $b_\ell$ is output.

For a vertex $v$, let the corresponding subcube $\clC_v$ be $\clC_v^{(1)} \times \ldots \times \clC_v^{(n)}$, where $\clC_v^{(i)}$ is a subcube of the domain of the $i$-th copy of $g$ (corresponding to the input $x^{(i)}$). Recall from Section~\ref{prelim} that for $b \in \{0,1\}$, $v_b$ denotes the $b$-th child of $v$.

For each leaf $\ell \in \mathcal{L}$ and $i=1, \ldots, n$, define $\snip^{(i)}(\ell)$ to be $1$ if there is a node $t$ in the unique path from the root of $\mathcal{B}$ to $\ell$ such that $\codim(\clC_t^{(i)})  < \D_\epsilon^\mu(g)$ and $\bias^\mu (\clC_t^{(i)}) \geq \frac{2}{n^2}$. Define $\snip^{(i)}(\ell)=0$ otherwise. Define $\snip(\ell)=\vee_{i=1}^n \snip^{(i)}(\ell)$.

For each $\ell \in \mathcal{L}$, define $p_\ell^z$ to be the probability that for an input drawn from $\gamma^z$, the computation of $\mathcal{B}$ terminates at leaf $\ell$. We have,
\begin{align}
\label{std1}
\Pr_{x \sim \gamma^z}[(x,\mathcal{B}(x)) \in f \circ g^n] = \Pr_{x \sim \gamma^z}[(z,\mathcal{B}(x)) \in f]=\sum_{\ell \in \mathcal{L}: (z,b_\ell)\in f} p_\ell^z.
\end{align}
From our assumption about $\mathcal{B}$ we also have that,
\begin{align}
\label{std2}
\Pr_{x \sim \gamma}[(x,\mathcal{B}(x)) \in f \circ g^n]=\underset{z \sim\lambda}{\mathbb{E}}\Pr_{x \sim \gamma^z}[(x,\mathcal{B}(x)) \in f \circ g^n] \geq \frac{2}{3}.
\end{align}
Now, consider a run of $\mathcal{A}'$ on $z$. For each $\ell \in \mathcal{L}$ of $\mathcal{B}$, define $q_\ell^z$ to be the probability that the computation of $\mathcal{A}'$ on $z$ terminates at leaf $\ell$ of $\mathcal{B}$. Note that the probability is over the internal randomness of $\mathcal{A}'$.

To finish the proof, we need the following two claims. The first one states that the leaves $\ell \in \mathcal{L}$ are sampled with similar probabilities by $\mathcal{B}$ and $\mathcal{A}'$.
\begin{claim}
\label{simileaf}
For each $\ell \in \mathcal{L}$ such that $\snip(\ell)=0$, and for each $z \in \{0,1\}^n$, $\frac{8}{9}\cdot p_\ell^z \leq q_\ell^z \leq \frac{10}{9} \cdot p_\ell^z$.
\end{claim}
The next Claim states that for each $z$, the probability according to $\gamma^z$ of the leaves $\ell$ for which $\snip(\ell)=1$ is small.
\begin{claim}
\label{lilsnip}
\[\forall z \in \{0,1\}^n, \sum_{\ell \in \mathcal{L}, \snip(\ell)=1} p_\ell^z \leq \frac{4}{n}.\]
\end{claim}
We first finish the proof of Theorem~\ref{thm:composing} assuming Claims~\ref{simileaf} and~\ref{lilsnip}, and then prove the claims. For a fixed input $z \in \{0,1\}^n$, the probability that $\mathcal{A}'$, when run on $z$, outputs an $r$ such that $(z,r) \in f$, is at least
\begin{align}
&\sum_{{\ell \in \mathcal{L},}\atop{(z,b_\ell)\in f, \snip(\ell)=0}} q_\ell^z \geq \sum_{{\ell \in \mathcal{L},}\atop{(z,b_\ell)\in f, \snip(\ell)=0}}\frac{8}{9} \cdot p_\ell^z\mbox{\ \ (By Claim~\ref{simileaf})} \nonumber \\
&=\frac{8}{9}\left(\sum_{{\ell \in \mathcal{L},}\atop{(z,b_\ell)\in f}} p_\ell^z-\sum_{{\ell \in \mathcal{L},}\atop{(z,b_\ell)\in f, \snip(\ell)=1}} p_\ell^z\right) \nonumber \\
& \geq \frac{8}{9}\left(\sum_{{\ell \in \mathcal{L},}\atop{(z,b_\ell)\in f}} p_\ell^z-\frac{4}{n} \right). \mbox{\ \ (By Claim~\ref{lilsnip})}\label{success}
\end{align}
Thus, the success probability of $\mathcal{A}'$ is at least
\begin{align*}
\underset{z \sim \lambda}{\mathbb{E}} \sum_{{\ell \in \mathcal{L},}\atop{(z,b_\ell)\in f, \snip(\ell)=0}} q_\ell^z &\geq \frac{8}{9}  \cdot \left(\underset{z \sim \lambda}{\mathbb{E}}\sum_{{\ell \in \mathcal{L},}\atop{(z,b_\ell)\in f}} p_\ell^z-\frac{4}{n} \right) \mbox{\ \ (By Equation~(\ref{success}))} \\
&\geq \frac{8}{9} \cdot \left(\frac{2}{3}-\frac{4}{n}\right) \mbox{\ \ (By Equations~(\ref{std1}) and~(\ref{std2}))} \\
&\geq \frac{5}{9}. \mbox{\ \ (For large enough $n$)}
\end{align*}
We now give the proofs of Claims~\ref{simileaf} and~\ref{lilsnip}.
\begin{proof}[Proof of Claim~\ref{simileaf}]
We will prove the first inequality. The proof of the second inequality is similar\footnote{Note that only the first inequality is used in the proof of Theorem~\ref{thm:composing}.}.

Fix a $z \in \{0,1\}^n$ and a leaf $\ell \in \mathcal{L}$. For each $i=1, \ldots, n$, assume that $\codim(\clC_\ell^{(i)})=d^{(i)}$, and in the path from the root of $\mathcal{B}$ to $\ell$ the variables $x^{(i)}_1, \ldots, x^{(i)}_{d^{(i)}}$ are set to bits $b_1, \ldots, b_{d^{(i)}}$ in this order. The computation of $\mathcal{A}'$ terminates at leaf $\ell$ if the values of the different bits $x^{(i)}_j$ sampled by $\mathcal{A}'$ agree with the leaf $\ell$. The probability of that happening is given by
\begin{align}
q_\ell^z=&\prod_{i=1}^n \Pr_{\mathcal{A}'}[x^{(i)}_1=b_1, \ldots, x^{(i)}_{d^{(i)}}=b_{d^{(i)}}\mid z] \\
&=\prod_{i=1}^n \Pr_{x \sim \mu}[x^{(i)}_1=b_1, \ldots, x^{(i)}_{\D^\mu_\epsilon(g)-1}=b_{\D^\mu_\epsilon(g)-1}]  \cdot \nonumber \\
&\qquad \qquad \qquad \Pr_{x \sim \mu_{z_i}}[x^{(i)}_{\D^\mu_\epsilon(g)}=b_{\D^\mu_\epsilon(g)}, \ldots, x^{(i)}_{d^{(i)}}=b_{d^{(i)}} \mid x^{(i)}_1=b_1, \ldots, x^{(i)}_{\D^\mu_\epsilon(g)-1}=b_{\D^\mu_\epsilon(g)-1}]. \label{one}
\end{align}
The second equality above follows from the observation that in Algorithm~\ref{A'}, the first $\D_\epsilon^\mu(g)-1$ bits of $x^{(i)}$ are sampled from their marginal distributions with respect to $\mu$, and the subsequent bits are sampled from their marginal distributions with respect to $\mu_{z_i}$. In equation~(\ref{one}), the term $\Pr_{x \sim \mu_{z_i}}[x^{(i)}_{\D^\mu_\epsilon(g)}=b_{\D^\mu_\epsilon(g)}, \ldots, x^{(i)}_{d^{(i)}}=b_{d^{(i)}} \mid x^{(i)}_1=b_1, \ldots, x^{(i)}_{\D^\mu_\epsilon(g)-1}=b_{\D^\mu_\epsilon(g)-1}]$ is interpreted as $1$ if $d^{(i)}<\D^\mu_\epsilon(g)$.

We invoke Claim~\ref{unbias}(b) with $\clC$ set to the subcube $\{x \in \{0,1\}^m:x^{(i)}_1=b_1, \ldots, x^{(i)}_{\D^\mu_\epsilon(g)-1}=b_{\D^\mu_\epsilon(g)-1}\}$ and $\delta$ set to $\frac{2}{n^2}$. To see that the claim is applicable here, note that from the assumption $\snip(\ell)=0$ we have that $\bias(\clC) < \delta=\frac{2}{n^2} < \frac{1}{2}$, where the last inequality holds for large enough $n$. Also, since $\D_\epsilon^\mu(g)>0$, by Proposition~\ref{fullbias} the bias of $\{0,1\}^m$ is at most $\frac{2}{n^4}<\frac{2}{n^2}=\delta$. Continuing from Equation~(\ref{one}), by invoking Claim~\ref{unbias}(b) we have,
\begin{align*}
q_\ell^z \geq &\prod_{i=1}^n (1-8/n^2) \Pr_{x \sim \mu_{z_i}}[x^{(i)}_1=b_1, \ldots, x^{(i)}_{\D^\mu_\epsilon(g)-1}=b_{\D^\mu_\epsilon(g)-1}]  \cdot \nonumber \\
&\qquad \qquad \qquad \Pr_{x \sim \mu_{z_i}}[x^{(i)}_{\D^\mu_\epsilon(g)}=b_{\D^\mu_\epsilon(g)}, \ldots, x^{(i)}_{d^{(i)}}=b_{d^{(i)}} \mid x^{(i)}_1=b_1, \ldots, x^{(i)}_{\D^\mu_\epsilon(g)-1}=b_{\D^\mu_\epsilon(g)-1}] \\
&=(1-8/n^2)^n \prod_{i=1}^n \Pr_{x \sim \mu_{z_i}}[x^{(i)}_1=b_1, \ldots, x^{(i)}_{d^{(i)}}=b_{d^{(i)}}] \\
&\geq \frac{8}{9} \cdot p_\ell^z .\mbox{\ \ (For large enough $n$)}
\end{align*}
\end{proof}
\begin{proof}[Proof of Claim~\ref{lilsnip}]
Fix a $z \in \{0,1\}^n$. We shall prove that for each $i$, $\sum_{\ell \in \mathcal{L}, \snip^{(i)}(\ell)=1} p_\ell^z \leq \frac{4}{n^2}$. That will prove the claim, since $\sum_{\ell \in \mathcal{L}, \snip(\ell)=1} p_\ell^z \leq \sum_{i=1}^n \sum_{\ell \in \mathcal{L}, \snip^{(i)}(\ell)=1} p_\ell^z$.

To this end, fix an $i \in \{1, \ldots, n\}$. For a random $x$ drawn from $\gamma^z$, let $p$ be the probability that in strictly less than $\D_\epsilon^\mu(g)$ queries the computation of $\mathcal{B}$ reaches a node $t$ such that $\bias(\clC_t^{(i)})$ is at least $\frac{2}{n^2}$. Note that this probability is over the choice of the different $x^{(j)}$'s. We shall show that $p \leq\frac{4}{n^2}$. This is equivalent to showing that $\sum_{\ell \in \mathcal{L}, \snip^{(i)}(\ell)=1} p_\ell^z \leq \frac{4}{n^2}$.

Note that each $x^{(j)}$ is independently distributed according to $\mu_{z_j}$. By averaging, there exists a choice of $x^{(j)}$ for each $j \neq i$ such that for a random $x^{(i)}$ chosen according to $\mu_{z_i}$, a node $t$ as above is reached within at most $\D^\mu_\epsilon(g)-1$ steps with probability at least $p$. Fix such a setting for each $x^{(j)}$, $j\neq i$. Claim~\ref{lilsnip} follows from Claim~\ref{R->bias} (note that $\epsilon=\frac{1}{2}-\frac{1}{n^4} \geq \frac{1}{4}$ for large enough $n$).
\end{proof}
This completes the proof of Theorem~\ref{thm:composing}. 
\end{proof}
\subsection{Hardness Amplification Using XOR Lemma}
\label{directsum}
In this section we prove Theorem~\ref{thm:xor}.

Theorem~\ref{thm:composing} is useful only when the function $g$ is hard against randomized query algorithms even for error $1/2-1/n^4$. In this section we use an XOR lemma to show a procedure that, given any $g$ that is hard against randomized query algorithms with error $1/3$, obtains another function on a slightly larger domain that is hard against randomized query algorithms with error $1/2-1/n^4$. This yields the proof of Theorem~\ref{thm:xor}.

Let $g:\{0,1\}^m \rightarrow \{0,1\}$ be a function. Let $g_t^\oplus: \left(\{0,1\}^m\right)^t \rightarrow \{0,1\}$ be defined as follows. For $x=(x^{(1)}, \ldots, x^{(t)}) \in \left(\{0,1\}^m\right)^t$,
\[g_t^\oplus(x) = \oplus_{i=1}^tg(x^{(i)}).\]
The following theorem is obtained by specializing Theorem $3$ of Andrew Drucker's paper \cite{DBLP:conf/coco/Drucker11} to this setting.
\begin{theorem}[Drucker 2011 \cite{DBLP:conf/coco/Drucker11} Theorem $3$]
\label{druker}
\label{ampl}
\[\R_{1/2-2^{-\Omega(t)}} (g_t^\oplus)=\Omega(t \cdot \R_{1/3}(g)).\]
\end{theorem}
Theorem~\ref{thm:xor} (restated below) follows by setting $t=\Theta(\log n)$ and combining Theorem~\ref{druker} with Theorem~\ref{thm:composing}.
\xor*
\paragraph{Acknowledgements:} This work was partially supported by the National Research Foundation, including under NRF RF Award No. NRF-NRFF2013-13, the Prime Minister’s Office, Singapore and the Ministry of Education, Singapore under the Research Centres of Excellence programme and by Grant No. MOE2012-T3-1- 009.\par D.G. is partially funded by the grant P202/12/G061 of GA \v CR and by RVO:\ 67985840. M. S. is partially funded by the ANR Blanc program under contract ANR-12-BS02-005 (RDAM project).
\bibliographystyle{plain}
\bibliography{ref}
\end{document}